\documentclass[11pt]{amsart}
\usepackage{amsmath, amssymb, amsbsy, amsfonts, amsthm, latexsym, amsopn, amstext, amsxtra, euscript, amscd, color, mathrsfs}
\usepackage{multirow, hyperref}

\usepackage{todonotes}
\usepackage{lipsum}
\usepackage{marginnote}

\makeatother

\newtheorem{thm}{Theorem}[section]

\newtheorem{cor}[thm]{Corollary}
\newtheorem{prop}[thm]{Proposition}

\newtheorem{rem}[thm]{Remark}

\newtheorem{rmk}[thm]{Remark}

\newtheorem{thm-con}[thm]{Theorem-Conjecture}
\numberwithin{equation}{section}

\theoremstyle{definition}

\def\cA{{\mathcal{A}}}
\def\cB{{\mathcal{B}}}
\def\cG{{\mathcal{G}}}
\def\cL{{\mathcal{L}}}

\def\Tr{{\rm Tr}}

\newcommand{\F}{\mathbb F}

\begin{document}
\title[Boomerang and extended differential uniformity]{A connection between the boomerang uniformity and the extended differential in odd characteristic and applications}

\author[M. Pal]{Mohit Pal}
\address{Department of Informatics, University of Bergen, PB 7803, N-5020, Bergen, Norway}
\email{mohit.pal@uib.no}

\author[P. St\u anic\u a]{Pantelimon St\u anic\u a}
\address{Applied Mathematics Department, Naval Postgraduate School, Monterey, CA 93943, USA}
\email{pstanica@nps.edu}

\maketitle

\begin{abstract}
This paper makes the first bridge between the classical differential/boomerang uniformity and the newly introduced $c$-differential uniformity. We show that the boomerang uniformity of an odd APN function is given by the maximum of the entries (except for the first row/column) of the function's $(-1)$-Difference Distribution Table. In fact, the boomerang uniformity of an odd permutation APN function equals its $(-1)$-differential uniformity. We then use this connection to easily compute the boomerang uniformity of several odd APN functions. In the second part we give two classes of differentially low-uniform functions obtained by modifying the inverse function. The first class of permutations (CCZ-inequivalent to the inverse) over a finite field $\F_{p^n}$ ($p$, an odd prime) is obtained from the composition of the inverse function with an order-$3$ cycle permutation, with differential uniformity $3$ if $p=3$ and $n$ is odd; $5$ if $p=13$ and $n$ is even; and $4$ otherwise. The second class is a family of binomials and we show that their differential uniformity equals~$4$. Finally, we extend to odd characteristic a result of Charpin and Kyureghyan (2010) providing an upper bound for the differential uniformity of the function and its switched version via a trace function. 
\end{abstract}

\section{Introduction}
Let $\F_{p^n}$ be the finite field with $p^n$ elements, where $p$ is an odd prime and $n$ is a positive integer. The set of nonzero elements of $\F_{p^n}$ forms a cyclic group with respect to multiplication and we shall denote it by $\F_{p^n}^*$. We shall denote by $\F_{p^n}[X]$, the ring of polynomials in the indeterminate $X$ and coefficients in $\F_{p^n}$. For any $\alpha \in \F_{p^n}$, we shall denote by $\chi(\alpha)$ the quadratic character of $\alpha$ and so, $\chi(\alpha)=0$ if $\alpha=0$, $\chi(\alpha)=1$ if $0\neq \alpha$ is a square, $\chi(\alpha)=-1$ if $\alpha$ is not a square. It is well-known, due to Lagrange's interpolation formula, that any function $f: \F_{p^n} \rightarrow \F_{p^n}$ can be uniquely expressed as a polynomial $f \in \F_{p^n}[X]/(X^{p^n}-X)$. A polynomial $f \in \F_{p^n}[X]$ is called a permutation polynomial (PP) if the induced mapping $c \mapsto f(c)$ permutes the elements of $\F_{p^n}$. The inverse map from $\F_{p^n}$ to itself, given by $X \mapsto X^{p^n-2}$, is an important class of functions due to its applications in coding theory, cryptography, among others. In fact, the inverse function over binary field $\F_{2^8}$ has been used as a substitution box in the block cipher AES, which is one of the most widely used cryptographic primitive.

For any function $f: \F_{p^n} \rightarrow \F_{p^n}$ and $a,b, c \in \F_{p^n}$, the $c$-Difference Distribution Table ($c$-DDT) entry at point $(a,b)$, denoted by ${_c}\Delta_f(a,b)$, is defined as ${_c}\Delta_f(a,b) := \lvert \{ X \in \F_{p^n} \mid f(X+a)-cf(X)=b \} \rvert.$ The $c$-differential uniformity ($c$-DU) of the function $f$, denoted by ${_c}\Delta_f$, is then defined as the maximum of ${_c}\Delta_f(a,b)$, where $a,b \in \F_{p^n}$ and $a\neq 0$ if $c=1$.  When ${_c}\Delta_f=1$ then $f$ is called perfect $c$-nonlinear (P$c$N) function (also known as $c$-planar function~\cite{BT20}) and when ${_c}\Delta_f=2$ then $f$ is called almost perfect $c$-nonlinear (AP$c$N) function. If $c=1$ then the notion of  the $c$-differential uniformity coincides with the classical notion of   differential uniformity. In the particular case when $f(X)=X^d$ for some positive integer $d$ then it is easy to see that ${_c}\Delta_f(a, b) = {_c}\Delta_f \left(1, \frac{b}{a^d} \right)$, for all $a \in \F_{p^n}^*$ and $b \in \F_{p^n}$. Therefore, for determining $c$-differential properties of a power map $f$, it is sufficient to consider the $c$-DDT entries with $a=1$ and $\gcd(d, p^n-1)$. Motivated by the notion of differential spectrum introduced by Blondeau et al.~\cite{BCC10}, Wang et. al~\cite{WZH22} introduced the notion of   the $c$-differential spectrum of power functions. For any power map $f(X)=X^d$ and $0 \leq i \leq {_c}\Delta_f$, let ${_c}\omega_i = \lvert \{b \in \F_{p^n} \mid {_c}\Delta_f(1,b)=i \} \rvert$ then the $c$-differential spectrum of $f$, denoted by ${_c}DS_f$, is defined as ${_c}DS_F := \{ {_c}\omega_i >0 \mid 0 \leq i \leq {_c}\Delta_f \}$. Again, if $c=1$ then the notion of   $c$-differential spectrum coincides with the classical notion of   differential spectrum.

To simplify the analysis of the boomerang attack~\cite{DW99}, which can be thought of as an extension of the differential attack, Cid et al.~\cite{cid18} introduced the notion of  boomerang connectivity table (BCT). In~\cite{cid18}, the BCT entries were defined for permutation functions in even characteristic, and their computation required the  inverse of the permutation.  In 2019, Li et al.~\cite{KLi19} gave an equivalent technique to compute BCT, which does not require the compositional inverse of the permutation polynomial $f(X)$ at all. For any $a, b \in \F_{p^n}$, the BCT entry of the function $f$ at point $(a, b)$, denoted by $\cB_f(a, b)$, is the number of solutions $(X,Y) \in \F_{p^n} \times \F_{p^n}$ of the following system of equations
\begin{equation*}
\begin{cases}
f(X) -f(Y) =b,\\
f(X+a) -f(Y+a)=b.
\end{cases}
\end{equation*}
To quantify the resistance of a function $f$ against the boomerang attack, Boura and Canteaut~\cite{BC18} coined the term boomerang uniformity, denoted by $\cB_f$, which is the maximum of $\cB_f(a,b)$, where $a,b \in \F_{p^n}^*$. In the particular case when $f(X)=X^d$ for some positive integer $d$ then $\cB_f(a,b) = \cB_f\left(1, \frac{b}{a^d} \right)$ for all $a, b \in \F_{p^n}^*$. Thus, for the power maps it is sufficient to consider the BCT entries with $a=1$. Analogous to the differential spectrum, the boomerang spectrum of a power map is defined in the following way. For any power map $f(X)=X^d$ and $0 \leq i \leq \cB_f$, let $v_i = \lvert \{b \in \F_{p^n}^* \mid \cB_f(1,b)=i \} \rvert$ then the boomerang spectrum of $f$, denoted by $BS_f$, is defined as $BS_f= \{ v_i >0 \mid 0 \leq i \leq \cB_f \}$. 

Though, so far, only one application of this new concept of the $c$-DU has been found in design theory~\cite{AKMRS23} (connecting the P$c$N property, when ${_c}\Delta_f=1$, to some quasigroups), we want to point out that in reality, for monomial functions $f(X)=X^d$, the differentials $\{f(\alpha X),f(X)\}$ used by Borissov et al.~\cite{Bor02} in their attack, are the same as the $c$-differentials at $a=0$, $\{f(X+a),cf(X) \}$, where $c=\alpha^d$. Further, in this paper, we shall establish a relation between BCT entries for odd APN functions and their $(-1)$-DDT entries. We then use this relation to derive two identities for the boomerang spectrum of odd APN functions. As an application of this result, we compute the boomerang spectrum of the inverse map from its $(-1)$-differential spectrum. By using our shown bridge connection between the BCT entries and and the $(-1)$-DDT entries, in addition to the inverse function, we compute the boomerang uniformity of four more classes of odd APN functions. It is worth mentioning here that the differential properties~\cite{ EFRST20, HRS99} and boomerang properties~\cite{JLLQ22} of the inverse function over finite fields of odd characteristic are known in the literature. For instance, when $c=0$ then the inverse function being a permutation is P$c$N. In the case of $c=1$, Helleseth et al.~\cite[Theorem 3]{HRS99} showed that $f$ is APN if $\chi(-3)=-1$, differentially $3$-uniform if $p=3$ and differentially $4$-uniform, otherwise. For $c \in \F_{p^n} \backslash \{0,1\}$, the $c$-differential uniformity $f$ has been considered in~\cite[Theorem 13]{EFRST20}, and the authors showed that if $c \in \F_{p^n} \backslash \{0,1, 4, 4^{-1}\}$ then the $c$-DU of $f$ is $3$ if $\chi(c^2-4c)=1$ or $\chi(1-4c)$; and $f$ is AP$c$N in all the remaining cases. Experimental results suggest that for $c \in \{4, 4^{-1}\}$, $f$ is not always AP$c$N. For instance, when $c=4$ then $f$ is differentially $3$-uniform over $\F_{17}$ and when $c=4^{-1}$ then $f$ is differentially $3$-uniform over $\F_{19}$. Here, we give a very simple proof correcting some conditions for the $c$-differential uniformity of $f$ for all the values of $c \in \F_{p^n}\backslash \{0,1\}$.

The second part of the paper is devoted to the construction of differentially low-uniform functions by modifying the inverse function. We know that the inverse map over finite fields of odd characteristic is an involution with three fixed points, namely $0$, $1$ and $-1$. We composed the inverse map with a $3$-length cycle $(0~1~-1)$ and showed that the resulting function $f(X)=X^{p^n-2} \circ (0~1~-1)$, which again is a permutation, has the differential uniformity
\[
\Delta_f=
\begin{cases}
    3~&~\mbox{if}~p=3~\mbox{and}~n~\mbox{is odd},\\
    5~&~\mbox{if}~p=13~\mbox{and}~n~\mbox{is even},\\
    4~&~\mbox{otherwise}.
\end{cases}
\]
In addition to it, we construct a class of differentially $4$-uniform binomials by adding the term $uX^2$ to the inverse map $X^{p^n-2}$. A well-known technique of constructing functions with low differential uniformity from the known ones is Dillon's switching method (see~\cite{Dillon09,EP09}). It was extended   by Budaghyan, Carlet and Leander~\cite{BCL09} (switching the Gold function $X^3$ to produce the APN function $X^3+\Tr(X^9)$), as well as Edel and Pott~\cite{EP09}. Carranza~\cite{Roberto20}  further extended the switching method to any differential uniformity, albeit in even characteristic. Here, we show a similar result in odd characteristic, and in particular, we show that if we switch the inverse map with $\alpha \Tr(h(X))$, where $\alpha \in \F_{p^n}^*$,  $\Tr$ is the absolute trace map and $h(X)\in \F_{p^n}[X]$, then the differential uniformity of the resulting function is bounded above by $2(p+1)$.

\section{Differential and boomerang properties of the inverse function}\label{S3}
In this section, we shall first establish a (perhaps, surprising, though not difficult to show) link between classical differential uniformity, generalized differential uniformity and boomerang uniformity. More precisely, we shall show that the BCT entries of odd APN functions and their $(-1)$-DDT entries have the following relation.
\begin{thm}
\label{TBUCDU}
Let $f$ be an odd APN function over $\F_{p^n}$, where $p$ an odd prime. Then, for any $a,b \in \F_{p^n}^*$ the BCT and the $(-1)$-DDT entries have the following relation
\[
\cB_f(a,b)= {_{-1}}\Delta_f(a,-b).
\]
\end{thm}
\begin{proof}
Recall that a function $f$ over $\F_{p^n}$ is said to be odd if and only if $f(-X)=-f(X)$ for all $X \in \F_{p^n}$. Since $f$ is an odd function, for any fixed $a \in \F_{p^n}^*$ and $b \in \F_{p^n}$ if $X$ is a solution of the equation  
\[
D_f(X, a) :=f\left (X+\frac{a}{2} \right ) - f\left (X-\frac{a}{2} \right ) =b,
\]
so is $-X$. Also since $f$ is APN, if solutions of the above equation exist then these two are the only solutions. Equivalently, for odd APN functions $D_f(X, a) = D_f(Y, a) \iff X = \pm Y$. Now recall that  the boomerang uniformity of the function $f$ is given by the maximum number of solutions of the following system of equations
 \begin{equation}\label{BCDU}
    \begin{cases}
    f\left (X-\frac{a}{2} \right )-f\left (Y-\frac{a}{2} \right ) =b\\
    f\left (X+\frac{a}{2} \right )-f\left (Y+\frac{a}{2} \right )=b
    \end{cases}
    \iff
     \begin{cases}
    f\left (X-\frac{a}{2} \right )-f\left (Y-\frac{a}{2} \right ) =b\\
    D_f (X, a )=D_f (Y, a ),
    \end{cases}
\end{equation}
where $a$ and $b$ are running over $\F_{p^n}^*$. Since $b \in \F_{p^n}^*$, $X=Y$ cannot be a solution of the first equation of the above system. Therefore, the only possibility is $X=-Y$ and in this case the first equation of the above system becomes
\[
f\left (Y+\frac{a}{2} \right )+f\left (Y-\frac{a}{2} \right ) =-b.
\]
Thus, $\cB_f(a,b) = {_{-1}}\Delta_f(a, -b)$ for all $a,b \in \F_{p^n}^*$. This completes the proof.
\end{proof}

\begin{rem}
A consequence of our previous proof is the fact that the boomerang uniformity of an odd function (non necessarily APN) in odd characteristic is influenced by the $(-1)$-differential uniformity. Thus, if one needs a low boomerang uniformity odd function, a necessary condition is that its $(-1)$-differential uniformity must be low. 
\end{rem}

\begin{rmk}
    In Theorem~\textup{\ref{TBUCDU}}, if $f$ is also a permutation, then ${_{-1}}\Delta_f(0,b)=1={_{-1}}\Delta_f(a,0)$, for all $a,b \in \F_{p^n}$. Thus, for odd APN permutations, the boomerang uniformity is equal to its $(-1)$-differential uniformity.
\end{rmk}

It is well-known that the $c$-differential spectrum entries of a power map $f(X)=X^d$ satisfy the following identities
\begin{equation}\label{I1}
    \sum_{i=0}^{{_c}\Delta_f} {_c}\omega_i =p^n~\quad \mbox{and}~\quad \sum_{i=0}^{{_c}\Delta_f} i \cdot {_c}\omega_i =p^n,
\end{equation}
which are useful in the computation of the $c$-differential spectrum. To the best of out knowledge, there are no such identities for the boomerang spectrum entries of a power map. Here, using Theorem~\ref{TBUCDU}, we give similar identities for the boomerang spectrum of odd APN power functions.

\begin{cor}
Let $f$ be an odd APN power function with boomerang uniformity $\cB_f$ and boomerang spectrum $\{v_i >0 ~|~ 0 \leq i \leq \cB_f \}$. Then the following identities hold:
\begin{equation}\label{I2}
    \sum_{i=0}^{\cB_f} v_i = p^n-1 ~\quad \mbox{and}~\quad \sum_{i=0}^{\cB_f} i \cdot v_i = p^n - {_{-1} }\Delta_f(1,0).
\end{equation}
\end{cor}
\begin{proof}
    The proof immediately follows from Theorem~\ref{TBUCDU}.
\end{proof}

In the remaining of this section, we shall show that Theorem~\ref{TBUCDU} is not only useful in the determination of the boomerang uniformity of odd APN functions but it can also be applied, in part, to determine the BCT entries of other differentially low-uniform odd functions. For instance, the inverse function $f(X)=X^{p^n-2}$ over finite fields $\F_{p^n}$ is APN if $\chi(-3)=-1$, differentially $3$-uniform if $p=3$ and differentially $4$-uniform, otherwise. We shall use Theorem~\ref{TBUCDU} to determine the boomerang spectrum of the inverse map in all these three cases. Before moving forward, we first prove the following theorem which gives a very simple proof correcting some conditions for the $c$-differential uniformity of $f$ for all the values of $c \in \F_{p^n}^*$.

\begin{thm}
\label{TD}
Let $f(X)=X^{p^n-2}$ be a function from $\F_{p^n}$ to itself and $c \in \F_{p^n}^*$. Then the $c$-differential uniformity of $f$ is
\[
\begin{cases}
3&~\mbox{if}~c \neq1, \chi(c^2-4c)=1~\mbox{or}~\chi(1-4c)=1,\\
3&~\mbox{if}~c=1~\mbox{and}~\chi(-3)=0,\\
4&~\mbox{if}~c=1~\mbox{and}~\chi(-3)=1,\\
2&~\mbox{otherwise}.
\end{cases}
\]
\end{thm}
\begin{proof}
Recall that the $c$-differential uniformity of $f(X)=X^{p^n-2}$ is given by the maximum number of solutions $X \in \F_{p^n}$ of the following equation 
\begin{equation}\label{T1E1}
    \left (X+\frac{1}{2} \right )^{p^n-2} - c \left (X-\frac{1}{2} \right )^{p^n-2} =b,
\end{equation}
where $b$ is running over $\F_{p^n}$. It is easy to see that if $\displaystyle X=\frac{1}{2}$ then $b=1$ and if $\displaystyle X=-\frac{1}{2}$ then $b=c$. When $\displaystyle X \not \in \left \{-\frac{1}{2}, \frac{1}{2} \right \}$, then Equation~\eqref{T1E1} reduces to 
\begin{equation}\label{T1E2}
  bX^2-(1-c)X+\frac{-b+2c+2}{4}=0.  
\end{equation}
If $b=0$ then we have no solution of Equation~\eqref{T1E2} if $c=1$ and a unique solution, otherwise. In the case when $b \neq 0$ then we have two solutions of Equation~\eqref{T1E2} if and only if
\[
\chi ((1-c)^2-b(-b+2c+2))=1.
\]
Thus, for all $b \not \in \{1,c \}$, we have at most two solutions of Equation~\eqref{T1E1}. It is easy to see that when $b=1$ then we have two solutions of Equation~\eqref{T1E2} if and only if $\chi(c^2-4c)=1.$ Similarly, when $b=c$ then we have two solutions of Equation~\eqref{T1E2} if and only if $\chi (1-4c)=1.$ This completes the proof.
\end{proof}

In~\cite[Theorems 3--6]{JLLQ22}, Jiang et al. determined the boomerang spectrum of the inverse function. In Theorem~\ref{thm:boom_inv}, we give a simpler proof for the boomerang spectrum utilizing the connection between BCT entries and $(-1)$-DDT entries. It is easy to see from Theorem~\ref{TD} that for the inverse function ${_{-1}}\Delta_f(1,0)=1$. Thus, for the APN inverse function (i.e., when $\chi(-3)=-1$), the boomerang spectrum is the same as the $(-1)$-differential spectrum with the only change that $v_1 = {_{-1}}\omega_1 -1$. The next theorem gives the $(-1)$-differential spectrum of the inverse map.
\begin{thm}
Let $f(X)=X^{p^n-2}$ be a function from $\F_{p^n}$ to itself. Then the $(-1)$-differential spectrum of $f$ is given by the following:
\begin{enumerate}
    \item If $p=3$ and $n$ is even, then
    \[
    \left \{{_{-1}}\omega_0=\frac{p^n-1}{2}, {_{-1}}\omega_1=3, {_{-1}}\omega_2=\frac{p^n-9}{2}, {_{-1}}\omega_3=2  \right \}.
    \]
    \item If $p=3$ and $n$ is odd, then
    \[
    \left \{{_{-1}}\omega_0=\frac{p^n-3}{2}, {_{-1}}\omega_1=3, {_{-1}}\omega_2=\frac{p^n-3}{2}  \right \}.
    \]
    \item If $p^n \equiv 3 \pmod 4$ and $\chi(5)=1$, then
    \[
    \left \{{_{-1}}\omega_0=\frac{p^n+1}{2}, {_{-1}}\omega_1=1, {_{-1}}\omega_2=\frac{p^n-7}{2}, {_{-1}}\omega_3=2  \right \}.
    \]
    \item If $p^n \equiv 3 \pmod 4$ and $\chi(5)=-1$, then
    \[
    \left \{{_{-1}}\omega_0=\frac{p^n-3}{2}, {_{-1}}\omega_1=3, {_{-1}}\omega_2=\frac{p^n-3}{2}  \right \}.
    \]
    \item If $p^n \equiv 1 \pmod 4$ and $\chi(5)=1$, then
    \[
    \left \{{_{-1}}\omega_0=\frac{p^n-1}{2}, {_{-1}}\omega_1=3, {_{-1}}\omega_2=\frac{p^n-9}{2}, {_{-1}}\omega_3=2  \right \}.
    \]
    \item If $p^n \equiv 1 \pmod 4$ and $\chi(5)=-1$, then
    \[
    \left \{{_{-1}}\omega_0=\frac{p^n-5}{2}, {_{-1}}\omega_1=5, {_{-1}}\omega_2=\frac{p^n-5}{2}  \right \}.
    \]
    \item If $p=5$, then
    \[
    \left \{{_{-1}}\omega_0=\frac{p^n-1}{2}, {_{-1}}\omega_1=1, {_{-1}}\omega_2=\frac{p^n-1}{2}  \right \}.
    \]
\end{enumerate}
\end{thm}
\begin{proof} We shall consider two cases, namely, $p=3$ and $p>3$. 

\noindent \textbf{Case 1.} Let $p=3$. In this case, from Theorem~\ref{TD}, we know that the $(-1)$-differential uniformity of $f$ is $3$ if $n$ is even and $2$ if $n$ is odd. Now consider the equation  
\begin{equation}\label{DSE1}
    (X+1)^{p^n-2} + (X-1)^{p^n-2} =b.
\end{equation} 
It is easy to observe that if $b=0$ then $X=0$ is the only solution of the above equation. Also, notice that, if $X=1$ then $b=-1$ and if $X=-1$ then $b=1$. When $X \not \in  \{-1, 1 \}$, then Equation~\eqref{DSE1} reduces to 
\begin{equation}\label{DSE2}
  bX^2+X-b=0.  
\end{equation}
Now, for $b \neq 0$,  Equation~\eqref{DSE2} has a unique solution if $b^2=-1$, which is possible if and only if $n$ is even. Thus, when $n$ is even then ${_{-1}}\omega_1=3$ and also for both $b=\pm 1$, Equation~\eqref{DSE2} has two solutions and hence ${_{-1}}\omega_3=2$. The remaining two entries of the $(-1)$-differential spectrum can be obtained using the identities~\eqref{I1}. Similarly, when $n$ is odd then ${_{-1}}\omega_1=3$ and the remaining two entries of the $(-1)$-differential spectrum can be obtained using the identities~\eqref{I1}.

\noindent \textbf{Case 2.} Let $p>3$. In this case, from Theorem~\ref{TD}, we know that the $(-1)$-differential uniformity of $f$ is $3$ if $\chi(5)=1$ and $2$, otherwise. Now consider the equation  
\begin{equation}\label{DSE3}
     \left (X+\frac{1}{2} \right )^{p^n-2} + \left (X-\frac{1}{2} \right )^{p^n-2} =b.
\end{equation} 
Again, if $b=0$ then $X=0$ is the only solution of Equation~\eqref{DSE3}. Also, it is easy to see that, if $\displaystyle X=\frac{1}{2}$, then $b=1$, and if $\displaystyle X=-\frac{1}{2}$, then $b=-1$. When $\displaystyle X \not \in \left \{-\frac{1}{2}, \frac{1}{2} \right \}$, then Equation~\eqref{DSE3} reduces to 
\begin{equation}\label{DSE4}
   bX^2-2X-\frac{b}{4}=0.  
\end{equation}
For $b \neq 0$, Equation~\eqref{DSE4} has a unique solution if $\chi(b^2+4)=0$, which is possible if and only if $p^n \equiv 1 \pmod 4$. Now, if $p=5$, then for $b = \pm 1$, Equation~\eqref{DSE4} has a unique solution and hence we have ${_{-1}}\omega_1=1$. The remaining two entries of the $(-1)$-differential spectrum can be obtained using the identities~\eqref{I1}. Similarly, if $\chi(5)=1$ then  ${_{-1}}\omega_1=3$ and $ {_{-1}}\omega_3=2$. Likewise, if $\chi(5)=-1$, then ${_{-1}}\omega_1=5$.
 
When $p^n \equiv 3 \pmod 4$, then $ {_{-1}}\omega_1=3$ if $\chi(5)=-1$. If $\chi(5)=1$, then $ {_{-1}}\omega_1=1$ and $ {_{-1}}\omega_3=2$. This completes the proof. 
\end{proof}

\begin{rem}
We note that Items $(1)$ and $(2)$ are special cases of Items $(5)$ and $(4)$, respectively. We prefer to give these cases separately, as they may be of interest.
\end{rem}

In the following theorem we consider the boomerang spectrum of the inverse function in the cases when it is not an APN function, i.e., when $\chi(-3) \in \{0,1\}$. This theorem exhibits how Theorem~\ref{TBUCDU} can be used to determine the boomerang spectrum of odd functions (not necessarily APN). Below, we use the notations
\[
Q_1=\chi \left(\frac{7-\sqrt{-3}}{2} \right )~\quad~\mbox{and}~\quad~Q_2= \chi \left(\frac{7+\sqrt{-3}}{2} \right ),
\]
in the understood finite fields.
\begin{thm}
\label{thm:boom_inv}
Let $f(X)=X^{p^n-2}$ be a function from $\F_{p^n}$ to itself and $\chi(-3) \in \{0,1\}$. 
Then the boomerang spectrum of $f$ is given by:
\begin{itemize}
\item If $p=3$, then
\begin{equation*}
    \begin{split}
        \left \{v_0=\frac{p^n-3}{2}, v_2=\frac{p^n-3}{2}, v_3=2  \right \}~\mbox{if}~n \equiv 1 \pmod 2,\\
        \left \{v_0=\frac{p^n-1}{2}, v_1=2, v_2=\frac{p^n-9}{2}, v_5=2  \right \}~\mbox{if}~n \equiv 0 \pmod 2.\\
    \end{split}
\end{equation*}

\item If $p=13$, then
\begin{equation*}
    \begin{split}
        \left \{v_0=\frac{p^n-9}{2}, v_1=2, v_2=\frac{p^n-1}{2}, v_3=2  \right \}~\mbox{if}~n \equiv 1 \pmod 2,\\
        \left \{v_0=\frac{p^n-1}{2}, v_2=\frac{p^n-13}{2}, v_3=4, v_4=2  \right \}~\mbox{if}~n \equiv 0 \pmod 2.\\
    \end{split}
\end{equation*}

\item If $\chi(-3)=1$, $p\neq 13$ and $\chi(5)=-1$, then
    \begin{enumerate}
    \item If $p^n \equiv 1 \pmod 4$, then 
    \begin{equation*}
    \begin{split}
        \left \{v_0=\frac{p^n-5}{2}, v_1=4, v_2=\frac{p^n-13}{2}, v_4=4  \right \}~\mbox{if}~Q_1=1=Q_2,\\
        \left \{v_0=\frac{p^n-9}{2}, v_1=4, v_2=\frac{p^n-5}{2}, v_4=2  \right \}~\mbox{if}~Q_1Q_2=-1,\\
        \left \{v_0=\frac{p^n-13}{2}, v_1=4, v_2=\frac{p^n+3}{2}  \right \}~\mbox{if}~Q_1=-1=Q_2.\\
    \end{split}
    \end{equation*}

    \item If $p^n \equiv 3 \pmod 4$, then 
    \begin{equation*}
    \begin{split}
        \left \{v_0=\frac{p^n-3}{2}, v_1=2, v_2=\frac{p^n-11}{2}, v_4=4  \right \}~\mbox{if}~Q_1=1=Q_2,\\
        \left \{v_0=\frac{p^n-7}{2}, v_1=2, v_2=\frac{p^n-3}{2}, v_4=2  \right \}~\mbox{if}~Q_1Q_2=-1,\\
        \left \{v_0=\frac{p^n-11}{2}, v_1=2, v_2=\frac{p^n+5}{2} \right \}~\mbox{if}~Q_1=-1=Q_2.\\
    \end{split}
    \end{equation*}
    \end{enumerate}

\item If $\chi(-3)=1$, $p\neq 13$ and $\chi(5)=1$, then
    \begin{enumerate}
    \item If $p^n \equiv 1 \pmod 4$, then 
    \begin{equation*}
    \begin{split}
        \left \{v_0=\frac{p^n-1}{2}, v_1=2, v_2=\frac{p^n-17}{2}, v_3=2 , v_4=4 \right \}~\mbox{if}~Q_1=1=Q_2,\\
        \left \{v_0=\frac{p^n-5}{2}, v_1=2, v_2=\frac{p^n-9}{2}, v_3=2, v_4=2  \right \}~\mbox{if}~Q_1Q_2=-1,\\
        \left \{v_0=\frac{p^n-9}{2}, v_1=2, v_2=\frac{p^n-1}{2}, v_3=2  \right \}~\mbox{if}~Q_1=-1=Q_2.\\
    \end{split}
    \end{equation*}

    \item If $p^n \equiv 3 \pmod 4$, then 
    \begin{equation*}
    \begin{split}
        \left \{v_0=\frac{p^n+1}{2}, v_2=\frac{p^n-15}{2}, v_3=2, v_4=4  \right \}~\mbox{if}~Q_1=1=Q_2,\\
        \left \{v_0=\frac{p^n-3}{2}, v_2=\frac{p^n-7}{2}, v_3=2, v_4=2  \right \}~\mbox{if}~Q_1Q_2=-1,\\
        \left \{v_0=\frac{p^n-7}{2}, v_2=\frac{p^n+1}{2}, v_3=2  \right \}~\mbox{if}~Q_1=-1=Q_2.\\
    \end{split}
    \end{equation*}
    \end{enumerate}
\item Let $\chi(-3)=1$ and $p=5$, which is equivalent to say that $p=5$ and $n$ is even. Then
     \begin{equation*}
    \begin{split}
        \left \{v_0=\frac{p^n-1}{2}, v_2=\frac{p^n-9}{2}, v_4=4  \right \}~\mbox{if}~Q_1=1=Q_2,\\
        \left \{v_0=\frac{p^n-9}{2}, v_2=\frac{p^n+7}{2}  \right \}~\mbox{if}~ Q_1=-1=Q_2.\\
    \end{split}
    \end{equation*}
\end{itemize}
\end{thm}
\begin{proof}
Recall that the boomerang uniformity of $f$ is given by the maximum number of solutions $(X,Y) \in \F_{p^n} \times \F_{p^n}$ of the following equation
\begin{equation}\label{BE1}
    \begin{cases}
    \left (X-\frac{1}{2} \right )^{p^n-2} - \left (Y-\frac{1}{2} \right )^{p^n-2} =b,\\
    D_f (X,1)= D_f (Y,1),
    \end{cases}
\end{equation}
where $b$ is running over $\F_{p^n}^*$. We shall now consider two cases, namely, $p=3$ and $\chi(-3)=1$, respectively.

\noindent \textbf{Case 1.} Let $p=3$. In this case System~\eqref{BE1} reduces to
\begin{equation}\label{BE2}
    \begin{cases}
    (X+1)^{p^n-2} - (Y+1)^{p^n-2} =b,\\
    (X+1)^{p^n-2} - (X-1)^{p^n-2} = (Y+1)^{p^n-2} - (Y-1)^{p^n-2}.
    \end{cases}
\end{equation}
When $D_f (X,1)= D_f (Y,1)=-1$, then after excluding solutions of the form $X=\pm Y$, we have a total of $4$ solutions of this equation, namely, $\{ (0,1), (0,-1), (1,0), (-1,0) \}$. Moreover, if $(X,Y) \in \{(0,1), (-1,0)\}$ then $b=-1$ and if $(X,Y) \in \{(0,-1), (1,0)\}$ then $b=1$. When $D_f (X,1)= D_f (Y,1) \neq -1$ then its solutions will be of the form $X=\pm Y$. Thus, for all $b \not \in \{0,1, -1 \}$, we have $\cB_f(1, b)={_{-1}}\Delta_f(1,-b)$ and when $b= \pm 1$ then $\cB_f(1, b)={_{-1}}\Delta_f(1,-b)+2$. Recall that 
\[
{_{-1}}\Delta_f(1,\pm 1) =
\begin{cases}
    1~&~\mbox{if}~n~\mbox{is odd},\\
    3~&~\mbox{if}~n~\mbox{is even}.
\end{cases}
\]
Thus, if $n$ is odd then $v_0={_{-1}}\omega_0$, $v_1={_{-1}}\omega_1-3$, $v_2={_{-1}}\omega_2$ and $v_3=2$. Similarly, when $n$ is even then $v_0={_{-1}}\omega_0$, $v_1={_{-1}}\omega_1-1$, $v_2={_{-1}}\omega_2$, $v_3={_{-1}}\omega_3-2$ and $v_5=2$.

\noindent \textbf{Case 2.} Let $\chi(-3)=1$ (which happens if $n$ is even or, $n$ is odd and $p\equiv 1\pmod 3$). In this case, if $D_f (X,1)= D_f (Y,1) \neq 1$ then from Theorem~\ref{TD}, the second equation of System~\eqref{BE1} has solutions $X= \pm Y$. When 
\begin{equation}\label{BE3}
    D_f (X,1)= D_f (Y,1)=1,
\end{equation}
then the solutions $(X,Y)$ of this equation and the corresponding value of $b$ from the first equation of System~\eqref{BE1}, i.e.,
\begin{equation}\label{BE4}
\left (X-\frac{1}{2} \right )^{p^n-2} - \left (Y-\frac{1}{2} \right )^{p^n-2} =b
\end{equation}
are given in Table~\ref{Btable1}. 
\begin{table}
\begin{center}
\begin{tabular}{ |l|c|c|c|c| } 
\hline
 & $Y=\frac{1}{2}$ & $Y=\frac{-1}{2}$ & $Y=\frac{\sqrt{-3}}{2}$ & $Y=\frac{-\sqrt{-3}}{2}$\\[2ex]
\hline
$X=\frac{1}{2}$ & $b=0$ & $b=1$ & $b=\frac{\sqrt{-3}+1}{2}$ & $b=\frac{-\sqrt{-3}+1}{2}$\\[1.5ex]
\hline
$X=\frac{-1}{2}$ & $b=-1$ & $b=0$ & $b=\frac{\sqrt{-3}-1}{2}$ & $b=\frac{-\sqrt{-3}-1}{2}$  \\[1.5ex]
\hline
$X=\frac{\sqrt{-3}}{2}$ & $b=\frac{-\sqrt{-3}-1}{2}$ & $b=\frac{-\sqrt{-3}+1}{2}$ & $b=0$ & $b=-\sqrt{-3}$\\[1.5ex]
\hline
$X=\frac{-\sqrt{-3}}{2}$ & $b=\frac{\sqrt{-3}-1}{2}$ & $b=\frac{\sqrt{-3}+1}{2}$ & $b=\sqrt{-3}$ & $b=0$\\[1.5ex]
\hline
\end{tabular}
\end{center}
\caption{Solutions $(X,Y)$ of Equation~\eqref{BE3} and corresponding values of $b$ from Equation~\eqref{BE4}.} 
\label{Btable1}
\end{table}
It is easy to observe, from Table~\ref{Btable1}, that the solutions $(X,Y)$ corresponding to $b \in\left \{ \pm 1, \pm \sqrt{-3}\right \}$ are of the form $X=-Y$. Thus, \ for all $b \in \F_{p^n}^* \backslash \cA$, where
\[
\cA := \left \{\pm  \left(\frac{\sqrt{-3}-1}{2}\right), \pm  \left(\frac{\sqrt{-3}+1}{2}\right) \right \},
\]
the BCT entries $\cB(1, b)$ are same as the $(-1)$-DDT entries ${_{-1}}\Delta_f(1,-b)$. We shall now consider two subcases, namely, $p=13$ and $p\neq 13$, respectively.\\

\noindent\textbf{Subcase 2.1.} Let $p=13$. In this case $\sqrt{-3}=7$ and hence the set $\cA$ becomes $\{ \pm 3, \pm 4 \}$. Now, we shall consider two cases, namely, $\chi(5)=-1$ and $\chi(5)=1$ which correspond to the cases of $n$ odd and $n$ even, respectively. If $n$ is odd, then for $b= \pm 4$, we have two solutions of Equation~\eqref{BE1} coming from Table~\ref{Btable1} and for $b = \pm 3$, we have three solutions of Equation~\eqref{BE1} and two among them are coming from Table~\ref{Btable1}. Thus, when $p=13$ and $n$ is odd then $v_0= {_{-1}}\omega_0-2$, $v_1= {_{-1}}\omega_1-3$, $v_2= {_{-1}}\omega_2+2$ and $v_3= 2$. Now, when $n$ is even then we have four solutions of Equation~\eqref{BE1} for $b=\pm 4$ and two of these are coming from Table~\ref{Btable1} and we have three solutions of Equation~\eqref{BE1} for $b = \pm 3$ and two of these are coming from Table~\ref{Btable1}. Thus, for $p=13$ and $n$ even, we have $v_0= {_{-1}}\omega_0$, $v_1= {_{-1}}\omega_1-3$, $v_2= {_{-1}}\omega_2-2$, $v_3= {_{-1}}\omega_3+2$ and $v_4=2$.\\

\noindent
\textbf{Subcase 2.2.} Let $p \neq 13$. Again, we shall consider three cases, namely, $\chi(5)=-1, \chi(5)=1$ and $\chi(5)=0$. If $\chi(5)=-1$, we have four solutions corresponding to $b = \pm  \left(\frac{\sqrt{-3}-1}{2}\right)$ if $\chi \left(\frac{7-\sqrt{-3}}{2} \right )=1$ and two solutions, otherwise. Similarly, we have four solutions corresponding to $b = \pm  \left(\frac{\sqrt{-3}+1}{2}\right)$ if $\chi \left(\frac{7+\sqrt{-3}}{2} \right )=1$ and two solutions, otherwise. Thus, if $\chi \left(\frac{7-\sqrt{-3}}{2} \right )=1=\chi \left(\frac{7+\sqrt{-3}}{2} \right )$, then
we have $v_0= {_{-1}\omega}_0$, $v_1= {_{-1}\omega}_1-1$, $v_2= {_{-1}\omega}_2-4$ and $v_4=4$. If $\chi \left(\frac{7-\sqrt{-3}}{2} \right) \cdot  \chi \left(\frac{7+\sqrt{-3}}{2} \right )=-1$, then we have $v_0= {_{-1}\omega}_0-2$, $v_1= {_{-1}\omega}_1-1$, $v_2= {_{-1}\omega}_2$ and $v_4=2$. If $\chi \left(\frac{7-\sqrt{-3}}{2} \right)=-1 = \chi \left(\frac{7+\sqrt{-3}}{2} \right )$, then we have $v_0= {_{-1}\omega}_0-4$, $v_1= {_{-1}\omega}_1-1$, $v_2= {_{-1}\omega}_2+4$.

If $\chi(5)=1$, then we have four solutions corresponding to $b = \pm  \left(\frac{\sqrt{-3}-1}{2}\right)$ if $\chi \left(\frac{7-\sqrt{-3}}{2} \right )=1$ and two solutions, otherwise. Similarly, we have four solutions corresponding to $b = \pm  \left(\frac{\sqrt{-3}+1}{2}\right)$ if $\chi \left(\frac{7+\sqrt{-3}}{2} \right )=1$ and two solutions, otherwise. Thus, if $\chi \left(\frac{7-\sqrt{-3}}{2} \right )=1=\chi \left(\frac{7+\sqrt{-3}}{2} \right )$, then
we have $v_0= {_{-1}\omega}_0$, $v_1= {_{-1}\omega}_1-1$, $v_2= {_{-1}\omega}_2-4$ and $v_4=4$. If $\chi \left(\frac{7-\sqrt{-3}}{2} \right) \cdot  \chi \left(\frac{7+\sqrt{-3}}{2} \right )=-1$, then we have $v_0= {_{-1}\omega}_0-2$, $v_1= {_{-1}\omega}_1-1$, $v_2= {_{-1}\omega}_2$ and $v_4=2$. If $\chi \left(\frac{7-\sqrt{-3}}{2} \right)=-1 = \chi \left(\frac{7+\sqrt{-3}}{2} \right )$, then we have $v_0= {_{-1}\omega}_0-4$, $v_1= {_{-1}\omega}_1-1$, $v_2= {_{-1}\omega}_2+4$.

If $\chi(5)=0$, then this together with the condition $\chi(-3)=1$ implies that $n$ is even. Now similar to the previous cases, we have four solutions corresponding to $b = \pm  \left(\frac{\sqrt{-3}-1}{2}\right)$ if $\chi \left(\frac{7-\sqrt{-3}}{2} \right )=1$ and two solutions, otherwise. Similarly, we have four solutions corresponding to $b = \pm  \left(\frac{\sqrt{-3}+1}{2}\right)$ if $\chi \left(\frac{7+\sqrt{-3}}{2} \right )=1$ and two solutions, otherwise. It is easy to observe that either both or none of $\chi \left(\frac{7-\sqrt{-3}}{2} \right )=1$ and $\chi \left(\frac{7+\sqrt{-3}}{2} \right )=1$, since 
\[
\chi \left(\frac{7-\sqrt{-3}}{2} \right ) \cdot \chi \left(\frac{7+\sqrt{-3}}{2} \right )= \chi(13) = \chi(3)= 1.
\]
Thus, if $\chi \left(\frac{7-\sqrt{-3}}{2} \right )=1=\chi \left(\frac{7+\sqrt{-3}}{2} \right )$, then we have $v_0= {_{-1}\omega}_0$, $v_1= {_{-1}\omega}_1-1$, $v_2= {_{-1}\omega}_2-4$ and $v_4=4$. If $\chi \left(\frac{7-\sqrt{-3}}{2} \right)=-1 = \chi \left(\frac{7+\sqrt{-3}}{2} \right )$, then we have $v_0= {_{-1}\omega}_0-4$, $v_1= {_{-1}\omega}_1-1$, $v_2= {_{-1}\omega}_2+4$. This completes the proof.
\end{proof}

In order to exemplify the usefulness of Theorem~\textup{\ref{TBUCDU}}, we shall now compute the boomerang uniformity of all the known classes of odd APN power functions over finite fields of odd characteristic. To the best of our knowledge, the following functions are  the only known classes of odd APN power maps $X^d$ over $\F_{p^n}$, for $p$ odd:
\begin{itemize}
    \item $f_1(X)=X^3$, $p\neq 3$~\cite{HRS99};
    \item $f_2(X)=X^{\frac{2 p^n - 1}{3}}$, $p^n \equiv 2 \pmod 3$~\cite{HRS99};
    \item $f_3(X)=X^{p^n-2}$, $p^n \equiv 2 \pmod 3$~\cite{HRS99};
    \item $f_4(X)=X^{p^{\frac{n}{2}}+2}$, $n$ even and $p^{\frac{n}{2}} \equiv 1 \pmod 3$~\cite{HRS99};
    \item $f_5(X)=X^{\frac{5^k+1}{2}}$, $p=5$, $\gcd(2n,k)=1$~\cite{HRS99};
    \item $f_6(X)=X^{\frac{5^n-1}{4}+\frac{5^{\frac{n+1}{2}}-1}{2}}$, $n$ odd~\cite{Dob03}.
\end{itemize}
It is easy to verify that the boomerang uniformity of $f_1$ is $3$. We know that under the given condition, the compositional inverse of $f_2$ is $X^3$ and since, in general, a permutation function and its compositional inverse share the same boomerang uniformity (see~\cite[Proposition 2]{BC18}), the boomerang uniformity of $f_2$ is also $3$. We already computed the boomerang uniformity of $f_3$. In the following theorem we shall compute the boomerang uniformity of the functions $f_4, f_5$ and $f_6$.

\begin{thm}
Let $f_4, f_5$ and $f_6$ be the functions on the finite fields $\F_{p^n}$ defined as above. Then $\cB_{f_4} \leq 5$ and $\cB_{f_5}=\cB_{f_6}=3$.
\end{thm}
\begin{proof}
We know, from Theorem~\ref{TBUCDU} that the boomerang uniformity of $f_4$ is given by the maximum number of solutions of the following equation
\[
\begin{split}
&\left (X+\frac{1}{2} \right )^{p^m+2} +\left (X-\frac{1}{2} \right )^{p^m+2} =b\\
\iff & 4X^{p^m+2} + X^{p^m} + 2X -2b=0,
\end{split}
\]
where $n=2m$ and $b$ is running over $\F_{p^n}^*$. To analyze its solutions we will be using Dobbertin's multivariate method. Let $Y=X^{p^m}$. The previous equation becomes 
\begin{align*}
&4X^2Y+Y+2X-2b=0,\text{ and raising it to the $p^m$ power},\\
& 4Y^2X+X+2Y-2b^{p^m}=0.
\end{align*}
If $b$ satisfies $4b^2+1=0$, then we get the solution $X=b$. If $4b^2+1\neq 0$, then we find $\displaystyle Y=\frac{2b-2X}{4X^2+1}$ from the first equation and replace it into the second equation arriving to
\begin{align*}
& 16 X^5-16b^{p^m} X^4  -8X^3 -16X^2 (1+b-b^{p^m}) -X(32b+3)+(16b^2+4b-2b^m)=0,
\end{align*}
which has at most five solutions. Thus $\cB_{f_4}\leq 5$ and experimental results for small values of $p$ and $n$ suggest that this bound is attained.

We now consider the boomerang uniformity of $f_5$. It was proved in~\cite[Theorem 6]{MRSYZ} that the $(-1)$-differential uniformity of the function $ X\mapsto X^{\frac{p^k+1}{2}}$ on $\F_{p^n}$ is $1$ if $\displaystyle \frac{2n}{\gcd(2n,k)}=1$, otherwise, it is $\displaystyle \frac{p^{\gcd(n,k)}+1}{2}$. Since $f_5$ on $\F_{5^n}$ is APN for $\gcd(2n,k)=1$, we therefore have that its $(-1)$-differential uniformity is exactly $\displaystyle \frac{5^{\gcd(n,k)}+1}{2}=3$ (since also $\gcd(n,k)=1$).

Finally, we shall consider the boomerang uniformity of $ f_6(X)=
X^{\frac{5^n-1}{4}+\frac{5^m-1}{2}}$, where $m=\frac{n+1}{2}$. Note that $f_6$ is a permutation since $\gcd(d,5^n-1)=1$, and  $\displaystyle \frac{5^m+1}{2} d\equiv 1\pmod {5^n-1}$. Thus, the boomerang uniformity of $f_6$ is equal to the boomerang uniformity (and hence $(-1)$-differential uniformity) of $f_7(X)=X^{\frac{5^m+1}{2}}$. From~\cite[Proposition~5]{S21}, we know that the $(-1)$-differential uniformity of $f_7$ is $3$. This completes the proof. 
\end{proof}

\section{Differentially low-uniform functions by modifying the inverse function}
The differential uniformity of functions over finite fields is preserved under certain transformations. For instance, let $f$ and $g$ are two functions over $\F_{p^n}$ such that $g= A_2 \circ f \circ A_1 +A$, for some affine permutations $A_1, A_2$ over $\F_{p^n}$ and some affine function $A$ over $\F_{p^n}$. Then, $f$ and $g$ have the same differential uniformity and we say that $f$ and $g$ are extended affine (EA) equivalent. The most general equivalence relation, known so far, which preserves the differential uniformity is the Carlet-Charpin-Zinoviev (CCZ) equivalence~\cite{CCZ98}. Two functions $f$ and $g$ over $\F_{p^n}$ are called CCZ-equivalent if there exists an affine permutation $\cA: \F_{p^n} \times \F_{p^n} \rightarrow \F_{p^n} \times \F_{p^n}$ which maps the graph $\cG_f:=(X,f(X))$ to the graph $\cG_g:=(X, g(X))$. Let $\cL$ be the linear part of the affine permutation $\cA$. Then~\cite[Lemma 3.1]{BCV20} shows that the affine permutation $\cA$ simply adds constants to input and output of the CCZ-equivalent function obtained by applying $\cL$. The CCZ-class of a function $f$ always contains the EA-class of the function $f$. It is well-known~\cite{BCP06} that if $f$ is a permutation then the CCZ-class also contains the EA-class of $f^{-1}$, the compositional inverse of the function $f$.

We know that the inverse map over finite fields of odd characteristic is an involution with three fixed points, namely $0,1$ and $-1$. In~\cite[Theorem 3.5]{JKK23}, the authors swapped the images of the inverse function at $0$ and $1$ and determined the $c$-differential uniformity of the function for all $c \in \F_p^*$. One may note that even after swapping the images of $0$ and $1$, this map has a fixed point $-1$. However, if we compose the inverse map by the $3$-length cycles $(0~ 1~ -1)$ or $(0~ -1~ 1)$, then it still remains a permutation with no fixed point. In the following theorem we shall determine the differential uniformity of the modified inverse function $f(X)= X^{p^n-2} \circ (0~ 1~ -1)$. Our results directly follow for the other modified map $f(X)= X^{p^n-2} \circ (0~ -1~ 1)$ as it is the compositional inverse of $f$.

\begin{thm}
Let $p>3$ be a prime number, $n$ be a positive integer and $f(X)= X^{p^n-2} \circ (0~ 1~ -1)$ be a map from $\F_{p^n}$ to itself. Then 
\begin{equation*}
    \begin{cases}
        \Delta_f= 5~&~\mbox{if}~p=13~\mbox{and}~n~\mbox{is even},\\
        \Delta_f \leq 4~&~\mbox{otherwise}.
    \end{cases}
\end{equation*}
\end{thm}
\begin{proof}  
We know that the differential uniformity of $f$ is given by the maximum number of solutions of the following equation
\begin{equation} \label{MIE1}
f\left(X +\frac{a}{2} \right) -f\left(X -\frac{a}{2} \right)= b,
\end{equation}
where $a,b \in \F_{p^n}$ and $a \neq 0$. Since $f$ is a permutation, if $b=0$, the above equation has no solutions, for all $a \in \F_{p^n}^*$. Now, we shall consider various cases depending upon the values of $\displaystyle X +\frac{a}{2}$ and $\displaystyle X -\frac{a}{2}$. More precisely, we shall consider the cases when $\displaystyle X -\frac{a}{2} \in \{0,1,-1\}$,  $\displaystyle X +\frac{a}{2} \in \{0,1,-1\}$ and $\displaystyle X \not \in \left \{\pm \frac{a}{2}, 1\pm \frac{a}{2}, -1 \pm \frac{a}{2} \right \}$, respectively.

\noindent \textbf{Case 1.} Let $\displaystyle X = \frac{a}{2}$. In this case, Equation~\eqref{MIE1} reduces to $f(a) -f(0)= b$, and so, $f(a)-1= b$.

\noindent \textbf{Case 2.} If $\displaystyle X = 1+\frac{a}{2}$, then from Equation~\eqref{MIE1}, $f(1+a)-f(1)=b$, that is, $f(1+a)+1= b$.

\noindent \textbf{Case 3.} Let $\displaystyle X = -1+\frac{a}{2}$. In this case, Equation~\eqref{MIE1} reduces to $f(-1+a)-f(-1)=b$, which is $f(-1+a)= b$.

\noindent \textbf{Case 4.} If $\displaystyle X = -\frac{a}{2}$, then from Equation~\eqref{MIE1}, $f(0)-f(-a)=b$, and so  $1-f(-a)= b$.

\noindent \textbf{Case 5.} Let $\displaystyle X = 1-\frac{a}{2}$. In this case, Equation~\eqref{MIE1} reduces to $f(1)-f(1-a)=b$, that is, $-1-f(1-a)= b$.

\noindent \textbf{Case 6.} If $\displaystyle X = -1-\frac{a}{2}$, then from Equation~\eqref{MIE1}, $f(-1)-f(-1-a)= b$, which is $-f(-1-a)= b$. 

\noindent \textbf{Case 7.} Let $\displaystyle X \not \in \left \{\pm \frac{a}{2}, 1\pm \frac{a}{2}, -1 \pm \frac{a}{2} \right \}$. Then Equation~\eqref{MIE1} reduces to 
\begin{equation} \label{MIE2}
\begin{split}
\left(X +\frac{a}{2} \right)^{-1} -\left(X -\frac{a}{2} \right)^{-1}= b \iff  X^2= \frac{a^2}{4} - \frac{a}{b}.
\end{split}
\end{equation}

One may note, from Cases 1--6, that the values of $b$ are in terms of some functions in the variable $a$. 
Now, in order to simplify the solutions $X$ and the corresponding values of $b$ from Cases 1--6, we consider five cases, namely, $a=1$, $a=-1$, $a=2$, $a=-2$ and $a \not \in \{\pm 1, \pm 2 \}$. This discussion is summarized in Table~\ref{MIT1}.
\begin{table}[h!]
\begin{center}
\begin{tabular}{ |c|c|c|c|c|c| } 
\hline
  & $a=1$ &  $a=-1$ & $a=2$ & $a=-2$& $ a\not \in \{ \pm1, \pm2 \}$\\[2ex]
\hline
Case 1 & $(\frac{1}{2}, -2)$ &   $(-\frac{1}{2}, -1)$ & $(1, -\frac{1}{2})$ & $(-1, -\frac{3}{2})$ & $(\frac{a}{2}, \frac{1}{a}-1)$  \\[1.5ex]
\hline
Case 2 & $(\frac{3}{2}, \frac{3}{2})$    & $(\frac{1}{2}, 2)$ & $(2, \frac{4}{3})$ & $(0,1)$ & $(1+\frac{a}{2},\frac{1}{1+a}+1)$  \\[1.5ex]
\hline
Case 3 & $ (-\frac{1}{2}, 1)$ &   $ (-\frac{3}{2}, -\frac{1}{2})$ & $(0,-1)$ & $(-2,-\frac{1}{3})$ & $(-1+\frac{a}{2}, \frac{1}{-1+a})$   \\[1.5ex]
\hline
Case 4 & $(-\frac{1}{2}, 1)$ &   $(\frac{1}{2}, 2)$ & $(-1, \frac{3}{2})$ & $(1, \frac{1}{2})$ &$(-\frac{a}{2},\frac{1}{a}+1)$     \\[1.5ex]
\hline
Case 5 & $(\frac{1}{2}, -2)$ &   $(\frac{3}{2}, -\frac{3}{2})$ & $(0,-1)$ & $(2, -\frac{4}{3})$ &$(1-\frac{a}{2}, \frac{1}{-1+a}-1)$   \\[1.5ex]
\hline
Case 6 & $( -\frac{3}{2}, \frac{1}{2})$ &   $ (-\frac{1}{2}, -1)$ & $(-2, \frac{1}{3})$& $(0,1)$ &$(-1-\frac{a}{2}, \frac{1}{1+a})$    \\[1.5ex]
\hline
Case 7 & $X^2= \frac{1}{4} - \frac{1}{b}$ &   $ X^2= \frac{1}{4} + \frac{1}{b}$ & $X^2= 1 - \frac{2}{b}$& $X^2= 1+ \frac{2}{b}$ & $X^2= \frac{a^2}{4} - \frac{a}{b}$   \\[1.5ex]
\hline
\end{tabular}
\end{center}
\caption{Pairs $(X,b)$ for different choices of $a$.}
\label{MIT1}
\end{table}

Now, we shall use Table~\ref{MIT1} to compute DDT entries for different values of $a$ and $b$. It is easy to observe from the column 2-5 of the Table~\ref{MIT1} that $\Delta_f(a,b) \leq 4$ for all $b \in \F_{p^n}$ and $a \in \{\pm 1, \pm 2 \}$. When $a \not \in \{\pm 1, \pm 2 \}$ then we can infer following from the Table~\ref{MIT1}

\begin{enumerate}
    \item We cannot have solutions from Case~1 and  Case~4 simultaneously as in this case $\frac{1}{a}-1 =\frac{1}{a}+1  \iff -1=1$, which is not possible as $p$ is odd.
    \item We cannot have solutions from Case~2 and Case~6 simultaneously, as $\frac{1}{1+a}+1=\frac{1}{1+a} \iff 1=0$, a contradiction.
    \item We cannot have solutions from Case~3 and Case~5 simultaneously, as $\frac{1}{-1+a}=-1+\frac{1}{-1+a} \iff 0=-1$, a contradiction.
    \item We cannot have solutions from Case~1 and Case~5 simultaneously, as $-1+\frac{1}{a}=-1+\frac{1}{-1+a} \iff 0=-1$, a contradiction.
    \item We cannot have solutions from Case~2 and Case~4 simultaneously, as $\frac{1}{1+a}+1=\frac{1}{a}+1 \iff 0=1$, a contradiction.
    \item We cannot have solutions from Case~3 and Case~6 simultaneously, as $\frac{1}{-1+a}=\frac{1}{1+a} \iff -1=1$, a contradiction.
\end{enumerate}
Thus, we have the following two possible scenarios in which we can get more than two solutions from Cases~1--6:

\begin{itemize}
        \item We now assume that we have solutions from Case~1, Case~2 and Case~3. Then $b=\frac{1}{a}-1 = \frac{1}{1+a}+1 =\frac{1}{-1+a}$. The second and third equalities produce the following system of equations
    \begin{equation*}
        \begin{cases}
            2a^2+2a-1=0,\\
            a^2=3,\\
            a^2-a+1=0.
        \end{cases}
    \end{equation*}
    One can easily verify that the above system of equations is consistent if and only if $p=13$ and $a=4$. Thus, for $p=13$ and $(a,b)=(4,9)$, we have the solutions $X=2$, $X=3$ and $X=1$ of Equation~\eqref{MIE1} from Case~1, Case~2 and Case~3, respectively. Now, for these parameters, the equation in the Case~7 becomes $X^2=5$, which has two solutions if $\chi(5)=1$ and no solutions, otherwise. We know that when $p=13$ then $\chi(5)=1$ if and only if $n$ is even. Thus,
    \[
    \Delta_f(4,9)=
    \begin{cases}
        3&~\mbox{if}~n~\mbox{is odd},\\
        5&~\mbox{if}~n~\mbox{is even}.\\
    \end{cases}
    \]

    \item We now assume that we have solutions from Cases~4--6. Then $b=\frac{1}{a}+1 =\frac{1}{-1+a}-1 = \frac{1}{1+a}$. The second and third equalities gives the following system of equations
    \begin{equation*}
        \begin{cases}
            2a^2-2a-1=0,\\
            a^2=3,\\
            a^2+a+1=0.
        \end{cases}
    \end{equation*}
    One can easily verity that the above system of equations is consistent if and only if $p=13$ and $a=9$. Thus, for $p=13$ and $(a,b)=(9,4)$, we have solutions $X=1$, $X=0$ and $X=-2$ of Equation~\eqref{MIE1} from Case~4, Case~5 and Case~6, respectively. Now, for these parameters, equation in the Case~7 becomes $X^2=5$, which has two solutions if $\chi(5)=1$ and no solutions, otherwise. We know that when $p=13$ then $\chi(5)=1$ if and only if $n$ is even. Thus,
    \[
    \Delta_f(9,4)=
    \begin{cases}
        3&~\mbox{if}~n~\mbox{is odd},\\
        5&~\mbox{if}~n~\mbox{is even}.\\
    \end{cases}
    \]
\end{itemize}
This completes the proof.
\end{proof}

\begin{rmk}
    Since the function $f$ in the above theorem is a permutation, it is CCZ-inequivalent to the inverse function, since for $p>3$ there is no permutation function in the CCZ-class of the inverse function. 
\end{rmk}

Computations revealed  that for some values of $p,n$ (like, $p=3, n=2,4$, or $p=5, n=3$ the differential uniformity of the function in our previous theorem is  $3$ (we could not find examples of uniformity lower than~$3$, though we have not performed extensive computations). We will show 
in our next result that if $p=3$, the differential uniformity of our function is indeed~$3$, for any odd~$n$.   

\begin{thm}
Let $n\geq 1$ be an integer and $f(X)= X^{3^n-2} \circ (0~ 1~ -1)$ be a map from $\F_{3^n}$ to itself. Then 
\begin{equation*}
    \Delta_f =
    \begin{cases}
        3~&~\mbox{if}~n~\mbox{is odd},\\
        4~&~\mbox{if}~n~\mbox{is even}.
    \end{cases}
\end{equation*}
\end{thm}
\begin{proof}  
We know that the differential uniformity of $f$ is given by the maximum number of solutions of the following equation
\begin{equation} \label{MIE31}
f(X-a) -f(X+a)= b,
\end{equation}
where $a,b \in \F_{p^n}$ and $a \neq 0$. Since $f$ is a permutation, if $b=0$ the above equation has no solutions for all $a \in \F_{p^n}^*$. We shall consider various cases depending upon the values of $X+a$ and $X-a$.

\noindent \textbf{Case 1.} Let $X =a$. In this case, Equation~\eqref{MIE31} reduces to $f(0)-f(2a)= b$, so $ 1-f(-a)= b$.

\noindent \textbf{Case 2.} If $X =1+a$, then from Equation~\eqref{MIE31} $f(1)-f(1+2a)=b$, thus $ -1-f(1-a)= b$.

\noindent \textbf{Case 3.} Let $X = -1+a$. In this case, Equation~\eqref{MIE31} reduces to $f(-1)-f(-1+2a)=b$, hence $ -f(-1-a)= b$.

\noindent \textbf{Case 4.} Let $X =-a$. In this case, Equation~\eqref{MIE31} reduces to $f(-2a)-f(0)= b$, thus $ f(a)-1= b$.

\noindent \textbf{Case 5.} If $X =1-a$, then from Equation~\eqref{MIE31} $f(1-2a)-f(1)=b$, so $ f(1+a)+1= b$.

\noindent \textbf{Case 6.} Let $X = -1-a$. In this case, Equation~\eqref{MIE31} reduces to $f(-1-2a)-f(-1)=b$, thus $ f(-1+a)= b$.

\noindent \textbf{Case 7.} Let $X \not \in \{\pm a, 1\pm a, -1 \pm a \}$. Then Equation~\eqref{MIE31} reduces to 
\begin{equation} \label{MIE31_2}
\begin{split}
(X -a)^{-1}-(X +a)^{-1}= b \iff  X^2=a^2 - \frac{a}{b}.
\end{split}
\end{equation}

One may note, from Cases 1--6, that the values of $b$ are in terms of some functions in the variable $a$.  Now, in order to simplify the solutions $X$ and corresponding values of $b$ from Cases 1--6, we consider five cases, namely, $a=1$, $a=-1$ and $a \not \in \{ 1,-1 \}$. This discussion is summarized in Table~\ref{MIT31}.\\
\begin{table}[h!]
\begin{center}
\begin{tabular}{ |c|c|c|c| } 
\hline
  & $a=1$ &  $a=-1$ & $ a\not \in \{ 1, -1 \}$\\[2ex]
\hline
Case 1 & $(1,1)$ &   $(-1, -1)$ & $(a, \frac{1+a}{a})$  \\[1.5ex]
\hline
Case 2 & $(-1,1)$    &  $(0,-1)$ & $(1+a,\frac{1+a}{1-a})$  \\[1.5ex]
\hline
Case 3 & $ (0, 1)$ &   $ (1, -1)$  & $(-1+a, \frac{1}{1+a})$   \\[1.5ex]
\hline
Case 4 & $(-1, 1)$ &   $(1, -1)$  &$(-a,\frac{1-a}{a})$     \\[1.5ex]
\hline
Case 5 & $(0, 1)$ &   $(-1, -1)$ &$(1-a, \frac{-1+a}{1+a})$   \\[1.5ex]
\hline
Case 6 & $( 1,1)$ &   $ (0, -1)$  &$(-1-a, \frac{1}{-1+a})$    \\[1.5ex]
\hline
Case 7 & $X^2= 1 - \frac{1}{b}$ &   $ X^2= 1 + \frac{1}{b}$   & $X^2= a^2 - \frac{a}{b}$   \\[1.5ex]
\hline
\end{tabular}
\end{center}
\caption{Pairs $(X,b)$ for different choices of $a$.}
\label{MIT31}
\end{table}

 It is easy to observe from Table~\eqref{MIT31} that the DDT entries 
 \[
 \Delta_f(1,b)=
 \begin{cases}
    3~&~\mbox{if}~b=1,\\
    2~&~\mbox{if}~b\neq -1~\mbox{and}~\chi(b(b-1))=1,\\
    0~&~\mbox{if}~b\neq -1~\mbox{and}~\chi(b(b-1))=-1.
 \end{cases}
 \]
Similarly,
\[
 \Delta_f(-1,b)=
 \begin{cases}
    3~&~\mbox{if}~b=-1,\\
    2~&~\mbox{if}~b\neq 1~\mbox{and}~\chi(b(b+1))=1,\\
    0~&~\mbox{if}~b\neq 1~\mbox{and}~\chi(b(b+1))=-1.
 \end{cases}
 \]
For $a \not \in \{0,1,-1\}$ we consider the following scenarios:

\begin{enumerate}
    \item We assume that we have a solution $X=a$ from the Case~1. Now, we cannot have solutions from Case~2 as in this case $b=\frac{1+a}{a}= \frac{1+a}{1-a}$ and the second equality would imply that $a=-1$, a contradiction. Also,  we cannot have solutions from Case~3 as the second equality of $b= \frac{1+a}{a}= \frac{1}{1+a}$ implies that $a=1$, a contradiction. Similarly, we cannot have solutions from Case~4, as  then $b=\frac{1+a}{a}= \frac{1-a}{a}$ and the second equality would imply that $a=0$, which again is a contradiction. Likewise, we cannot have solutions from Case~5, since then $b=\frac{1+a}{a}= \frac{-1+a}{1+a}$ and the second equality implies that $1=0$. Let us assume that we have a solution from Case~6, and so, $b=\frac{1+a}{a}= \frac{1}{-1+a}$, so $ b=a$ and $a^2=a+1$, and in this case the equation in Case~7 becomes $X^2=a$. Thus if $b=a$ and $a^2=a+1$, then we have 
    \[
    \Delta(a,b)=
    \begin{cases}
        2~&~\mbox{if}~n~\mbox{is odd},\\
        2~&~\mbox{if}~n~\mbox{is even and}~\chi(a)=-1,\\
        4~&~\mbox{if}~n~\mbox{is even and}~\chi(a)=1.
    \end{cases}
    \]

    \item We now assume that we have a solution $X=1+a$ from Case~2. We have already seen that we cannot have solutions from Case~1. Now, we cannot have solutions from Case~3, as in this case $b=\frac{1+a}{1-a}= \frac{1}{1+a}$ and the second equality would imply that $a=0$, a contradiction. Similarly, we cannot have a solution from Case~4, as in this case $b=\frac{1+a}{1-a} = \frac{1-a}{a}$ and the second equality implies that $a=1$, a contradiction. Likewise, we cannot have a solution from Case~6, as the second equality of $b=\frac{1+a}{1-a}= \frac{1}{-1+a}$ implies that $a=1$, which is a contradiction. Now, let us assume that we have a solution from Case~5 then $b=\frac{1+a}{1-a}=\frac{-1+a}{1+a}$, so $ b=a$ and $a^2=-1$. Notice that for $b=a$ with $a^2=-1$, the equation in Case~7 reduces to $X^2=1$, which always has two solutions $X=\pm 1$. Thus if $b=a$ and $a^2=-1$, then we have 
    \[
    \Delta(a,b)=
    \begin{cases}
        2~&~\mbox{if}~n~\mbox{is odd},\\
        4~&~\mbox{if}~n~\mbox{is even}.
    \end{cases}
    \]

    \item We now assume that we have solutions from Case~3. We have already seen that we cannot have solutions from Case~1 and Case~2. It is easy to verify that if we  have a solution from Case~5, the second equality of $b=\frac{1}{1+a}=\frac{-1+a}{1+a}$ implies that $a=-1$, a contradiction. Similarly, if we  have a solution from Case~6, then $b=\frac{1}{1+a}=\frac{1}{-1+a}$ and the second equality implies that $1=-1$, which again is a contradiction. Now, let us assume that we have a solution from Case~4 then $b=\frac{1}{1+a}=\frac{1-a}{a}$, so $ b= a$ and $a^2=1-a$. It is easy to observe that for $b=a$ and $a^2=1-a$, the equation in Case~7 reduces to $X^2=-a$. Thus, if $b=a$ and $a^2=1-a$, then we have 
    \[
    \Delta(a,b)=
    \begin{cases}
        2~&~\mbox{if}~n~\mbox{is odd},\\
        2~&~\mbox{if}~n~\mbox{is even and}~\chi(-a)=-1,\\
        4~&~\mbox{if}~n~\mbox{is even and}~\chi(-a)=1.
    \end{cases}
    \]
    
    \item We now assume that we have a solution from Case~4. We have already seen that in this case we  cannot have solution from Case~1 and Case~2. We have also discussed the case when we have solutions from Case~4 and Case~3, simultaneously. One may note that if we  have a solution from Case~5, then $b=\frac{1-a}{a}=\frac{-1+a}{1+a}$ and the second equality implies that $a=1$, which is a contradiction. Similarly, we  have solution from Case~6 as the second equality of $b=\frac{1-a}{a}=\frac{1}{-1+a}$ implies that $a=-1$.

    \item We now assume that we have a solution from Case~5. We have already shown that in this case we  have solutions from Cases~1, 3 and 4. Also, we have already discussed the case when we have solutions from Case~5 and Case~2, simultaneously. One can easily verify that if we  have solutions from Case~5 and Case~6 simultaneously, then $\frac{-1+a}{1+a}=\frac{1}{-1+a}$, so $ a=0$, a contradiction.
\end{enumerate}
This completes the proof.
\end{proof}

Constructing low differential uniform functions from known ones is a common theme in many works. In the same research vein, in the following result, we shall show that if we add a monomial term $uX^2$ to the inverse mapping then it still remains differentially $4$-uniform. However by doing so, it is no longer a permutation. Later, we will construct new functions with low differential uniformity, which may remain permutations (and we provide examples of such).

\begin{prop}
Let $f(X)=X^{p^n-2}+uX^2$, where $u \in \F_{p^n}^*$, be a function from $\F_{p^n}$ to itself. Then the differential uniformity of $f$ is $\leq 4$.
\end{prop}
\begin{proof}
We know that the differential uniformity of $f$ is given by the maximum number of solutions of the following equation
\begin{equation} \label{PIE1}
\left(X +\frac{a}{2} \right)^{p^n-2} -\left(X -\frac{a}{2} \right)^{p^n-2} +2au X= b,
\end{equation}
where $a,b \in \F_{p^n}$ and $a \neq 0$. It is straightforward to see that if $\displaystyle X= \frac{a}{2}$ then $b= a^{-1}+a^2u$ and if $\displaystyle X= -\frac{a}{2}$ then $b= a^{-1}-ua^2$. When $\displaystyle X \not \in \{-\frac{a}{2}, \frac{a}{2} \} $ then Equation~\eqref{PIE1} reduces to 
\begin{equation} \label{PIE2}
2auX^3-bX^2-\frac{a^3u}{2}X +\frac{ba^2}{4}-a=0,
\end{equation}
which can have at most $3$ solutions in $\F_{p^n}$. This completes the proof.
\end{proof}

Charpin-Kyureghyan~\cite[Proposition 3]{CK10} showed that the differential uniformity of $G(X)=F(X)+\gamma\Tr(H(X))$ is upper bounded by twice the differential uniformity of~$F$, that is, $\Delta_G\leq 2\Delta_F$, in even characteristic. In the following result, we shall show that a similar result holds for odd characteristic, as well, and in the particular  case of switching the inverse function we obtain a stronger result.

\begin{prop}
Let $f,g$ be defined on $\F_{p^n}$, and  $g(X)=f(X) + \alpha\Tr(h(X))$, where $\alpha\in\F_{p^n}$ and $\Tr: \F_{p^n} \rightarrow \F_p$ is the absolute trace function. Then the differential uniformity of $g$, $\Delta_g\leq p\cdot \Delta_f$. Further, if $f(X)=X^{p^n-2}$, then  $\Delta_g\leq 2(p+1)$.
\end{prop}
\begin{proof}
We know that the differential uniformity of $g$ is given by the maximum number of solutions, for $a\neq 0,b\in\F_{p^n}$, of the following equation,  
\begin{equation} \label{PICE1}
\begin{split}
g\left(X +a \right)-g\left(X \right)=f(X+a)-f(X)+\alpha\Tr \left (h(X+a)-h(X) \right) &= b.
\end{split}
\end{equation}
Since $\Tr \left (h(X+a)-h(X) \right)=\epsilon\in\F_p$, an argument similar as the one of~\cite{CK10} shows that
\[
\displaystyle \Delta_g(a,b)\leq \sum_{i=0}^{p-1} \Delta_f(a,b-i\alpha),
\]
from which we can infer that $\Delta_g\leq p \cdot \Delta_f$.

If $f(X)=X^{p^n-2}$ then $f$ is APN, if $p^n\equiv 2\pmod 3$, has differential uniformity $3$ when $p=3$, and $4$ in all other cases. Our argument  below is only better for the inverse than the general one in the case when $f$ is not APN. We write the differential equation slightly differently, namely,
$\displaystyle D_g(X,a):=g\left(X +\frac{a}{2} \right)-\left(X -\frac{a}{2} \right)=b$, more precisely,
\begin{equation} \label{PICE2}
\begin{split}
\left(X +\frac{a}{2} \right)^{p^n-2} -\left(X -\frac{a}{2} \right)^{p^n-2} +\Tr\left(h\left(x+\frac{a}{2}\right)-h\left(x-\frac{a}{2}\right) \right) &= b,\text{ that is}\\
\left(X +\frac{a}{2} \right)^{p^n-2} -\left(X -\frac{a}{2} \right)^{p^n-2} +\alpha\Tr \left (D_h(X,a)\right ) &= b.
\end{split}
\end{equation}
We shall now consider two cases, namely, $\displaystyle X=\pm \frac{a}{2}$  and $\displaystyle X \not \in \left \{ \frac{a}{2}, -\frac{a}{2} \right \}$. Notice that if $\displaystyle X= \pm \frac{a}{2}$, then Equation~\eqref{PICE2} reduces to $a^{p^n-2}+\alpha\Tr \left(D_h\left(\pm \frac{a}{2},a\right) \right) = b.$ When $\displaystyle X \not \in \left \{ \frac{a}{2}, -\frac{a}{2} \right \}$, then  Equation~\eqref{PICE2} reduces to
\begin{equation}\label{CC1}
     \alpha\left(X^2 -\frac{a^2}{4} \right) \Tr \left(D_h\left(X,a\right) \right)= b \left(X^2 -\frac{a^2}{4} \right) +a,
\end{equation}
that is,
\[
\Tr \left(D_h\left(X,a\right) \right)=\frac{b}{\alpha}+\frac{a}{\alpha \left(X^2 -\frac{a^2}{4} \right)}\in\F_p.
\]
Now, let $\displaystyle X^2 -\frac{a^2}{4} := Z^{-1}$. Then using the properties of the trace function, observe that
\[
\left(\frac{aZ+b}{\alpha}\right)^p -\left(\frac{aZ+b}{\alpha}\right) =0.
\]
We conclude that $\displaystyle Z=\frac{\alpha \zeta-b}{a}$, for any $\zeta\in\F_p$. Surely, we cannot claim that the number of solutions for Equation~\eqref{CC1} is precisely $p$, since the above value for $Z$, rendering $2p$ values of $X$ may not all satisfy the original differential equation~\eqref{PICE1} for $g$.
\end{proof}

\begin{rem}
In the above proposition, if $f$ is such that the derivative traces
$\Tr \left(D_f\left(-\frac{a}{2},a\right) \right) \neq \Tr \left(D_f\left(\frac{a}{2},a\right) \right) $, then the bound of the differential uniformity of the switched inverse function becomes $2p+1$.
\end{rem}

While this is not a systematic computation, we can surely do a switching of  the inverse function to obtain permutation polynomials that preserve the differential uniformity of the inverse. In particular, we can find permutation APN functions for some small dimensions, easily. 
We took functions $f$ of the form $f(X,d,s)=X^{p^n-2}+\Tr\left(g^s X^d\right)$ ($g$ is a primitive element of the underlying finite field).
We tabulate below some computational data for small primes and dimensions (we only list the permutation polynomials, where $d\leq p-1, s\leq p-1$, which preserve the differential uniformity (DU) of the inverse function, surely, the case of $(d,s)=(0,0)$; we also removed the trivial cases of $d=0$ and $s>0$).
\begin{center}
\begin{tabular}{ |c|c|c| } 
\hline
$(p,n)$ & $(d,s)$ & DU\\
\hline
$(3,2)$ & $(0,0),(5,1),(7,1),(4,2),(5,2),(7,2)$ & $3$\\
\hline
$(3,3)$ & $(0,0),(13,0), (17,0), (23,0), (25,0), (13,1), (17,1), (23,1), (25,1) $ & $3$\\
\hline
$(5,2)$ & $(0,0), (19,0), (23,0), (6,3), (18,3), (19,3), (23,3), (19,4), (23,4)$ & $4$\\
\hline
$(5,3)$ & $(0,0), (99,0),(119,0), (123,0), (31,1), (62,1),(93,1),(99,1)$& \\
& $ (119,1), (123,1), (99,3), (119,3), (123,3), (99,4),(119,4), (123,4)$ & $2$\ (APN)\\
\hline
$(7,2)$ & $(0,0),(41,0), (47,0), (41,1), (47,1), (41,2), (47,2), (8,4),(16,4),$ & \\
& $(24,4),(32,4),(40,4)(41,4),(47,4), (41,5), (47,5), (41,6), (47,6) $ & $4$\\
\hline
\end{tabular}
\end{center}

 \section{Conclusions}
 
 In this paper  we start by correcting some conditions on the $c$-DU of the inverse function for $c \in \F_q \ \{0,1\}$, and give two identities concerning the boomerang spectrum of a function.
We next show that a necessary condition for a low boomerang uniformity of an odd function is for the $(-1)$-differential uniformity to be low, as well. In fact, in the case of odd APN permutations, they are equal. We apply this result to find the boomerang spectrum of the inverse function and the boomerang uniformity of four other odd APN functions. 
Moreover, we find a new class of differentially $\le 4$-uniform permutations in characteristic $p\neq 13$ (respectively, differentially 5-uniform when $p=13$) that is CCZ-inequivalent to the inverse function. Finally, we provide an upper bound for the differential uniformity of a switched function, thus extending a result of Charpin and Kyureghyan~\cite{CK10} to odd characteristic.

\section*{Acknowledgements}
The research of Mohit Pal is supported by the Research Council of Norway under Grant No. 314395. Pantelimon St\u anic\u a thanks the Selmer Center at the University of Bergen for the invitation to visit, and for the excellent working conditions while this paper was started.

\section*{Declarations}

\textbf{Conflict of interest} The authors declare that they have no conflict of interest regarding the publication of this paper.


\begin{thebibliography}{99}
\bibitem{AKMRS23} N. Anbar, T. Kalayci, W. Meidl, C. Riera, P. St\u anic\u a, {\it P{$\wp$}N functions, complete mappings and quasigroup difference sets}, J. Combin. Designs \textbf{31} (2023), 667--690.

\bibitem{BT20} D. Bartoli, M. Timpanella, {\it On a generalization of planar functions}, J. Algebra Comb. \textbf{52} (2020), 187--213.

\bibitem{BCC10}  C. Blondeau, A. Canteaut, P. Charpin, {\em Differential properties of power functions}, Int. J. Inf. Coding Theory \textbf{1}(2) (2010), 149--170.

\bibitem{BCC11}  C. Blondeau, A. Canteaut, P. Charpin, {\it Differential properties of $x \mapsto x^{2^t-1}$}, IEEE Trans. Inf. Theory \textbf{57}(12) (2011), 8127--8137.

\bibitem{BC18} C. Boura, A. Canteaut, {\it On the boomerang uniformity of cryptographic Sboxes,} IACR Trans. Symmetric Cryptol. \textbf{2018}(3) (2018), 290--310.

\bibitem{BCV20} L. Budaghyan, M. Calderini, I. Villa, {\it On relations between CCZ- and EA-equivalences}, Cryptogr. Commun. \textbf{12} (2020), 85--100.

\bibitem{BCL09} L. Budaghyan, C. Carlet, G. Leander, {\it Constructing new APN functions from known ones}, Finite Fields Appl. 15 (2009), 150--159.

\bibitem{BCP06} L. Budaghyan, C. Carlet, A. Pott, {\it New classes of almost bent and almost perfect nonlinear polynomials}, IEEE Trans. Inf. Theory, \textbf{52}(3) (2006), 1141--1152.

\bibitem{Bor02} N. Borisov, M. Chew, R. Johnson, D. Wagner, {\it Multiplicative differentials}, In: J. Daemen, V. Rijmen (eds) Fast Software Encryption. FSE 2002, LNCS 2365, Springer, Berlin, Heidelberg, 2002.


\bibitem{CCZ98} C. Carlet, P. Charpin, V. Zinoviev, {\it Codes, bent functions and permutations suitable for DES-like cryptosystems}, Des. Codes Cryptgr. \textbf{15} (1998), 125--156.

\bibitem{Roberto20} R. C. R. Carranza, {\it Construction of new differentially  $\delta$--uniform families}, Ph.D. Dissertation, University of Puerto Rico, Rio Piedras, 2020; available at \url{https://repositorio.upr.edu/bitstream/handle/11721/2378/UPRRP_MATE_ReyesCarranza_2020.pdf?sequence=1&isAllowed=y}.

\bibitem{CK10} P. Charpin, G. Kyureghyan, {\it Monomial functions with linear structure and permutation polynomials}, In: Finite Fields: Theory and Applications, Contemp. Math. 518, 3, (16) Amer. Math. Soc., 2010, pp. 99--111.

\bibitem{cid18} C. Cid, T. Huang, T. Peyrin, Y. Sasaki, L. Song, {\it Boomerang connectivity table: a new cryptanalysis tool,} In: J. Nielsen, V. Rijmen (eds.) Adv. in Crypt.-EUROCRYPT 2018, LNCS 10821, pp. 683--714. Springer, Cham (2018).

\bibitem{Dillon09} J. F. Dillon, {\it APN polynomials: an update},  International Conf. on Finite Fields and Applic. -- Fq9, 2009.

\bibitem{Dob03} H. Dobbertin, D. Mills, E. N. Muller, A. Pott, W. Willems, {\it APN functions in odd characteristic}, Discr. Math. \textbf{267} (2003), 95--112.

\bibitem{EP09} Y. Edel, A. Pott, {\em A new almost perfect nonlinear function which is not quadratic}, Adv. Math. Commun. \textbf{3}(1) (2009), 59--81.

\bibitem{EFRST20} P. Ellingsen, P. Felke, C. Riera, P. St\u anic\u a, A. Tkachenko, {\it $C$-differentials, multiplicative uniformity and (almost) perfect $c$-nonlinearity}, IEEE Trans. Inf. Theory \textbf{66}(9) (2020), 5781--5789.


\bibitem{HRS99} T. Helleseth, C. Rong, D. Sandberg, {\it New families of almost perfect nonlinear power functions}, IEEE Trans. Inf. Theory \textbf{45} (1999), 475--485.

\bibitem{JKK23} J. Jeong, N. Koo, S. Kwon, {\it Low $c$-differential uniformity of the swapped inverse function in odd characteristic}, Discret. Appl. Math. 336 (2023), 195--209.

\bibitem{JLLQ22} S. Jiang, K. Li, Y. Li, L. Qu, {\it Differential and boomerang spectrums of some power permutations}, Cryptogr. Commun. \textbf{14} (2022), 371--393.


\bibitem{KLi19} K. Li, L. Qu, B. Sun, C. Li, {\it New results about the boomerang uniformity of permutation polynomials,} IEEE Trans. Inform. Theory 65(11) (2019), 7542--7553.

\bibitem{MRSYZ} S. Mesnager, C. Riera, P. St\u anic\u a, H. Yan, Z. Zhou, {\em Investigation on c-(almost) perfect nonlinear functions}, IEEE Trans. Inf. Theory {\bf 67}(10) (2021), 6916--6925.

\bibitem{S21} P. St\u anic\u a, {\em Investigations on c-Boomerang Uniformity and Perfect Nonlinearity}, Discrete Applied Mathematics 304 (2021), 297--314.

\bibitem{WZH22} X. Wang, D. Zheng, L. Hu, {\it Several classes of P$c$N power functions over finite fields,} Discret. Appl. Math. 322 (2022), 171--182.

\bibitem{DW99} D. Wagner, {\it The boomerang attack,} In: Knudsen, L.R. (ed.) Fast Software Encryption-FSE 1999. LNCS 1636, Springer, Berlin, Heidelberg, pp. 156--170 (1999).
\end{thebibliography}
\end{document}